\documentclass[runningheads]{llncs}
\usepackage[utf8]{inputenc}
\usepackage[T1]{fontenc}
\usepackage{xspace}
\usepackage{graphicx}
\usepackage{amsfonts,amsmath,amssymb}
\usepackage{algorithm}
\usepackage[noend]{algpseudocode}
\usepackage{url}
\usepackage{hyperref}
\usepackage{fullpage}

\newcommand{\silbot}{${\sf SILBOT}$\xspace}
\newcommand{\algo}{{\textit{WRain}}\xspace}

\newcommand{\W}{$\mathbb{W}$\xspace}
\newcommand{\E}{$\mathbb{E}$\xspace}
\newcommand{\NE}{$\mathbb{NE}$\xspace}
\newcommand{\NW}{$\mathbb{NW}$\xspace}
\newcommand{\SE}{$\mathbb{SE}$\xspace}
\newcommand{\SW}{$\mathbb{SW}$\xspace}

\newcommand{\async}{{\sc Async}\xspace}
\newcommand{\cont}{\sc{contracted}\xspace}
\newcommand{\expd}{\sc{expanded}\xspace}
\begin{document}

\title{Asynchronous Silent Programmable Matter:\\ Line Formation\thanks{A brief announcement about the results contained in this paper appears in the proceedings of the 37th International Symposium on Distributed Computing (DISC) 2023~\cite{NPbf23}.\\ The work has been supported in part by the Italian National Group for Scientific Computation (GNCS-INdAM).}}

\titlerunning{Asynchronous Silent Programmable Matter: Line Formation}
\author{Alfredo~Navarra\inst{1}\orcidID{0000-0001-8547-5934} \and Francesco~Piselli\inst{1}
}
\authorrunning{A.~Navarra, F.~Piselli}
\institute{Dipartimento di Matematica e Informatica, Universit\`a degli Studi di Perugia, Italy. \email{alfredo.navarra@unipg.it}; \email{francesco.piselli@unifi.it}}
\maketitle 
\begin{abstract}
Programmable Matter (PM) has been widely investigated in recent years. It refers to some kind of matter with the ability to change its physical properties (e.g., shape or color) in a programmable way. One reference model is certainly Amoebot, with its recent canonical version (DISC 2021). Along this line, with the aim of simplification and to better 
address concurrency, the \silbot model has been introduced (AAMAS 2020), which heavily reduces the available capabilities of the particles composing the PM. 
In \silbot, in fact, particles
are asynchronous, without any direct means of communication (silent) and without memory of past events (oblivious). 
Within \silbot, we consider the \emph{Line formation} primitive in which particles are required to end up in a configuration where they are all aligned and connected. We propose a simple and elegant distributed algorithm -- optimal in terms of number of movements, along with its correctness proof. 

\end{abstract}
%
\section{Introduction}
\label{sec:intro}
The design of smart systems intended to adapt and organize themselves in order to accomplish global tasks is receiving more and more interest, especially with the technological advance in nanotechnology, synthetic biology and smart materials, just to mention a few. Among such systems, main attention has been devoted in the recent years to the so-called \emph{Programmable Matter} (PM). This refers to some kind of matter with the ability to change its physical properties (e.g., shape or color) in a programmable way. PM can be realized by means of weak self-organizing computational entities, called \emph{particles}. 

In the early 90s, the interest in PM by the scientific community was mostly theoretical. In fact, the ideas arising within such a context did not find support in technology that was unprepared for building computational devices at micro/nanoscale. Nowadays, instead, nano-technology has greatly advanced and the pioneering ideas on PM could find a practical realization. The production of nano units that integrate computing, sensing, actuation, and some form of motion mechanism are becoming more and more promising. Hence, the investigation into the computational characteristics of PM systems has assumed again a central role, driven by the applied perspective. In fact, systems based on PM can find a plethora of natural applications in many different contexts, including smart materials, ubiquitous computing, repairing at microscopic scale, and tools for minimally invasive surgery. Nevertheless, the investigation on modeling issues for effective algorithm design, performance analysis and study on the feasibility of foundational tasks for PM have assumed a central and challenging role.
Various models have been proposed so far for PM. One that deserves main attention is certainly Amoebot, introduced in~\cite{DGRSS14}. By then, various papers have considered that model, possibly varying some parameters. Moreover, a recent proposal to try to homogenize the referred literature has appeared in~\cite{DRS21}, with the intent to enhance the model with concurrency. 

One of the weakest models for PM, which includes concurrency and eliminates direct communication among particles as well as local and shared memory, is \silbot~\cite{DDDNP20}. The aim has been to investigate on the minimal settings for PM under which it is possible to accomplish basic global tasks in a distributed fashion. Actually, with respect to the Amoebot model, in \silbot particles admit a 2 hops distance visibility instead of just 1 hop distance. Even though this does not seem a generalization of \silbot with respect to Amoebot, the information that can be obtained by means of communications (and memory) in Amoebot may concern particles that are very far apart from each other. Moreover, there are tasks whose resolution has been shown to require just 1 hop distance visibility even in \silbot (see, e.g.~\cite{NPBT23}), perhaps manipulating some other parameters. Toward this direction of simplification and in order to understand the requirements of basic tasks within PM, we aim at studying in \silbot the \emph{Line formation} problem, where particles are required to reach a configuration where they are all aligned (i.e., lie on a same axis) and connected.

\subsection{Related work}
The relevance of the Line formation problem is provided by the interest shown in the last decades within various contexts of distributed computing.
In graph theory, the problem has been considered in~\cite{GRSST14} where the requirement was to design a distributed algorithm that, given an arbitrary connected graph $G$ of nodes with unique labels, converts $G$ into a sorted list of nodes.
In swarm robotics, the problem has been faced from a practical point of view, see, e.g.~\cite{JL14}. The relevance of line or V-shape formations has been addressed in various practical scenarios, as in~\cite{CSG23,C19,YWB18}, based also on nature observation. In fact, ants form lines for foraging activities whereas birds fly in V-shape in order to reduce the air resistance. In robotics, line or V-shape formations might be useful for exploration, surveillance or protection activities. 
Most of the work on robots considers direct communications, memory, and some computational power. For application underwater or in the outerspace, instead, direct communications are rather unfeasible and this motivates the investigation on removing such a capability, see, e.g.~\cite{JW18,SDS15}.
Concerning more theoretical models, the aim has been usually to study the minimal settings under which it is possible to realize basic primitives like Line formation. In~\cite{CGKH21,SC20}, for instance, Line formation has been investigated for (semi-)synchronized robots (punctiform or not, i.e., entities occupying some space) moving within the Euclidean plane, admitting limited visibility, and sharing the knowledge of one axis on direction.
For synchronous robots moving in 3D space, in~\cite{YUKY17}, the plane formation has been considered, which might be considered as the problem corresponding to Line formation for robots moving in 2D.
In~\cite{pairbot2020}, 
robots operate within a triangular grid and Line formation is required as a preliminary step for accomplishing the Coating  of an object. The environment as well as the movements of those robots remind PM. Within Amoebot, Line formation has been approached in~\cite{DGS15}, subject to the resolution of Leader Election, which is based, in turn, on communications and not on movements.
%
\subsection{Outline}
In the next section, we provide all the necessary definitions and notation, along with the formalization of the Line formation problem. In Section~\ref{sec:prel}, we give some preliminary results about the impossibility to resolve Line formation within \silbot. Then, in Section~\ref{sec:algo}, we provide a resolution algorithm for the case of particles sharing a common orientation. In Section~\ref{sec:run}, we show a possible running example about the proposed algorithm. In Section~\ref{sec:corr}, we prove the correctness as well as the optimality in terms of number of moves of the proposed algorithm. Finally, in Section~\ref{sec:concl}, we provide some conclusive remarks and possible directions for future work.

\section{Definitions and notation}
In this section, we review the \silbot model for PM introduced in~\cite{DDDNP20b,DDDNP20}, 
and then we formalize the Line formation problem along with other useful definitions.

In \silbot, particles operate on an infinite triangular grid embedded in the plane. Each node can contain at most one particle. 
Each particle is an automaton with two states, {\cont} or {\expd} (they do not have any other form of persistent memory). In the former state, a particle occupies a single node of the grid while in the latter, the particle occupies one single node and one of the adjacent edges, see, e.g. Figure~\ref{fig:es1}. Hence, a particle always occupies one node, at any time. 
Each particle can sense its surrounding up to a distance of $2$ hops, i.e., if a particle occupies a node $v$, then it can see the neighbors of $v$, denoted by $N(v)$, and the neighbors of the neighbors of $v$. Hence, within its visibility range, a particle can detect empty nodes, {\cont}, and {\expd} particles.

Any positioning of {\cont} or {\expd} particles that includes all $n$ particles composing the system is referred to as a \emph{configuration}. 
Particles alternate between active and inactive periods decided by an adversarial schedule, independently for each particle.

In order to move, a particle alternates between {\expd} and {\cont} states. In particular, a {\cont} particle occupying node $v$ can move to a neighboring node $u$ by expanding along edge $(v,u)$, and then re-contracting on $u$. Note that, if node $u$ is already occupied by another particle then the {\expd} one will reach $u$ only if $u$ becomes empty, eventually, in a successive activation. There might be arbitrary delays between the actions of these two particles. When the particle at node $u$ has moved to another node, the edge between $v$ and $u$ is still occupied by the originally {\expd} particle. In this case, we say that node $u$ is \emph{semi-occupied}. 

\noindent {\em A particle commits itself into moving to node $u$ by expanding in that direction. At the next activation of the same particle, it is constrained to move to node $u$, if $u$ is empty. A particle cannot revoke its expansion once committed.}

The \silbot model introduces a fine grained notion of asynchrony with possible delays between observations and movements performed by the particles. 
This reminds 
the so-called \async schedule designed for theoretical models dealing with mobile and oblivious robots (see, e.g.~\cite{CDN21a,DDNNS15,FPS-macbook19}). 
All operations performed by the particles are non-atomic: there can be delays between the actions of sensing the surroundings, 
computing the next decision (e.g., expansion or contraction), executing the decision. 

The well-established fairness assumption is included, where each particle must be activated within finite time, infinitely often, in any execution of the particle system, see, e.g.,~\cite{FPS-macbook19}.

Particles are required to take deterministic decisions. Each particle may be activated at any time independently from the others. Once activated, a particle looks at its surrounding (i.e., at 2 hops distance) and, on the basis of such an observation, decides (deterministically) its next \textit{action}.

If two {\cont} particles decide to expand on the same edge simultaneously, exactly one of them (arbitrarily chosen by the adversary) succeeds. 

If two particles are {\expd} along two distinct edges incident to a same node $w$, toward $w$, and both particles are activated simultaneously, exactly one of the particles (again, chosen arbitrarily by the adversary) contracts to node $w$, whereas the other particle does not change its {\expd} state according to the commitment constraint described above.

\begin{figure}[t]
        \centering      
  a) \includegraphics[width=4cm]{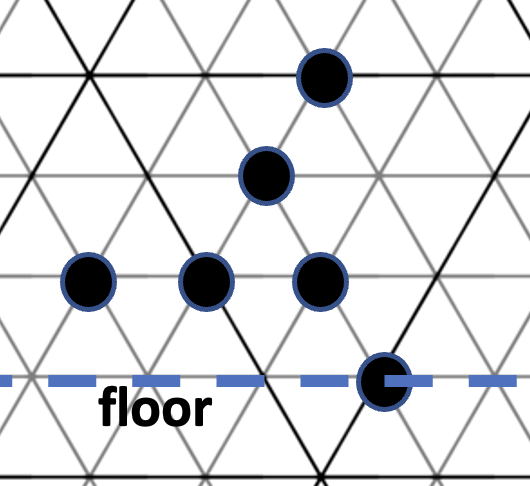}
  		\hspace{5mm}        
  b) \includegraphics[width=4cm]{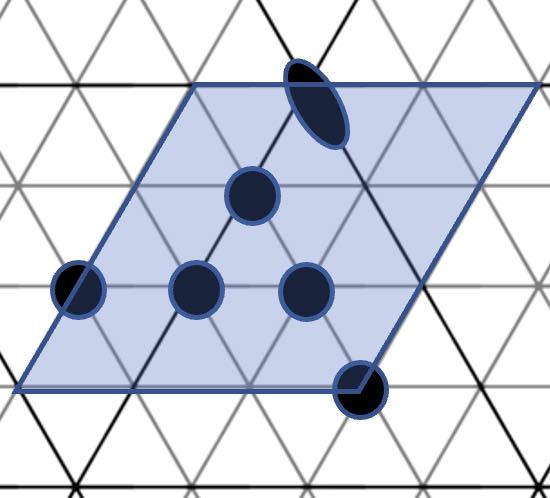}
        \caption{\small(\textit{a}) 
        A possible initial configuration with emphasized the \emph{floor} (dashed line); 
                        (\textit{b}) 
        a possible evolution of the configuration shown in (a) with an expanded particle. The shaded parallelogram is the minimum bounding box containing all the particles.}
        \label{fig:es1}
\end{figure}
\begin{figure}[t]
        \centering
  a) \includegraphics[width=3cm]{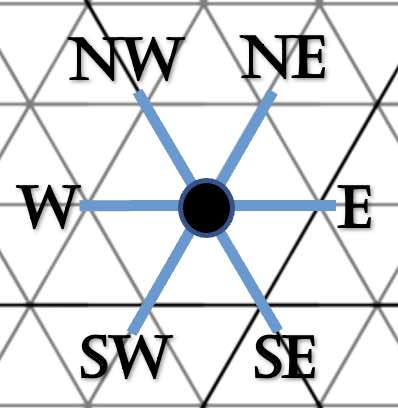}	
    	\hspace{5mm}        
  b) \includegraphics[width=3cm]{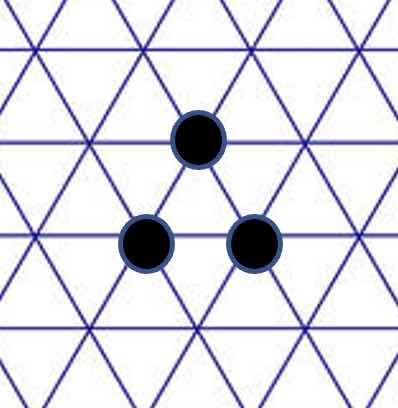}
		\hspace{5mm}        
  c) \includegraphics[width=3.6cm]{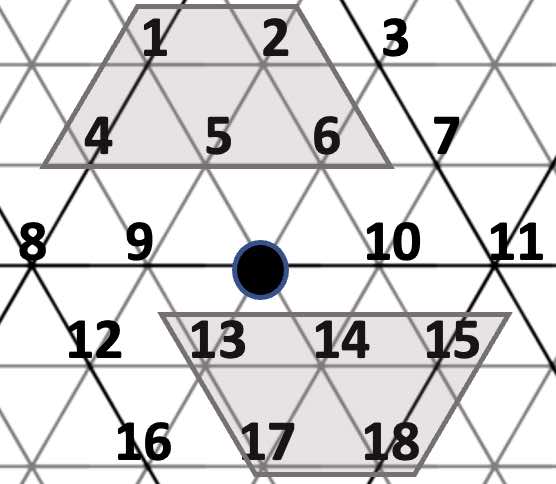}
        \caption{\small(\textit{a}) 
        A representation of the orientation of a particle;
        (\textit{b}) An initial configuration where Line formation is unsolvable within \silbot;
        (\textit{c}) Enumerated visible neighborhood of a particle; the two trapezoids emphasize two relevant areas for the definition of the resolution algorithm.}
        \label{fig:orientation}
\end{figure}

A relevant property that is usually required in such systems concerns connectivity. 
A configuration is said to be \emph{connected} if the set of nodes occupied by particles induce a connected subgraph of the grid.

\begin{definition}\label{def:initial}
A configuration is said to be \emph{initial}, if all the particles are {\cont} and connected.
\end{definition}

\begin{definition}\label{def:linear}[Line formation]
Given an initial configuration, the \emph{Line formation} problem asks for an algorithm that leads to a configuration where all the particles are {\cont}, connected and aligned.
\end{definition}

\begin{definition}
Given a configuration $C$, the corresponding \emph{bounding box} of $C$ is the smallest parallelogram with sides parallel to the West--East and SouthWest--NorthEast directions, enclosing all the particles.
\end{definition}

See Figure~\ref{fig:es1}.b for a visualization of the bounding box of a configuration. Note that, in general, since we are dealing with triangular grids, there might be three different bounding boxes according to the choice of two directions out of the three available. As it will be clarified later, for our purposes we just need to define one by choosing the West--East and SouthWest--NorthEast directions. In fact, as we are going to see in the next section, in order to solve Line formation in \silbot, we need to add some capabilities to the particles. In particular, we add a common orientation to the particles. As shown in Figure~\ref{fig:orientation}.a, all particles commonly distinguish among the six directions of the neighborhood that by convention are referred to as the cardinal points \NW, \NE, \W, \E, \SW, and \SE. 

Furthermore, in order to describe our resolution algorithm, we need two further definitions that identify where the particles will be aligned.

\begin{definition}
Given a configuration $C$, the line of the triangular grid containing the southern side of the bounding box of $C$ is called the \emph{floor}.
\end{definition}

\begin{definition}\label{def:final}
A configuration is said to be \emph{final} if all the particles are {\cont}, connected and lie on floor.
\end{definition}

By the above definition, a final configuration is also initial.
Moreover, if a configuration is final, then Line formation has been solved. Actually, it might be the case that a configuration satisfies the conditions of Definition~\ref{def:linear} but still it is not final with respect to Definition~\ref{def:final}. This is just due to the design of our algorithm that always leads to solve Line formation on floor. 

\section{Impossibility results}\label{sec:prel}
\begin{table}[t]
\center
\begin{tabular}{lllcc}
\textit{Problem} & \textit{Schedule} & \textit{View} & \textit{Orientation} & \textit{Reference}\\ \hline
Leader Election \ \  & \async              & 2 hops                 & no                   & \cite{DDDNP20b}      \\
Scattering       & ED-\async \ \          & 1 hop                  & no                   & \cite{NPBT23}       \\
Coating          & \async              & 2 hops                 & chirality            & \cite{NP23}    \\
Line formation	 & \async				  & 2 hops				   & yes					  & \bf{this paper}\\ \hline
\end{tabular}
\caption{Literature on \silbot.}\label{tab:problems}
\end{table}
As shown in the previous section, the \silbot model is very constrained in terms of particles capabilities. Since its first appearance~\cite{DDDNP20}, where the Leader Election problem has been solved, the authors pointed out the need of new assumptions in order to allow the resolution of other basic primitives. In fact, due to the very constrained capabilities of the particles, it was not possible to exploit the election of a leader to solve subsequent tasks. The parameters that can be manipulated have concerned the type of schedule, the hop distance from which particles acquire information, and the orientation of the particles. Table~\ref{tab:problems} summarizes the primitives so far approached within \silbot and the corresponding assumptions. Leader Election was the first problem solved when introducing \silbot~\cite{DDDNP20b}. Successively, the Scattering problem has been investigated~\cite{NPBT23}. It asks for moving the particles in order to reach a configuration where no two particles are neighboring to each other. Scattering has been solved by reducing the visibility range to just 1 hop distance but relaxing on the schedule which is not \async. In fact, the ED-\async schedule has been considered. It stands for Event-Driven Asynchrony, i.e., a particle activates as soon as it admits a neighboring particle, even though all subsequent actions may take different but finite time as in \async. For Coating~\cite{NP23}, where particles are required to surround an object that occupies some connected nodes of the grid, the original setting has been considered apart for admitting chirality, i.e., a common handedness among particles.

In this paper, we consider the Line formation problem, where particles are required to reach a configuration where they are all aligned and connected. About the assumptions, we add a common orientation to the particles to the basic \silbot model. The motivation for endowing the particles with such a capability comes by the following result:

\begin{theorem}
\label{th:imposs}
Line formation is unsolvable within \silbot, even though particles share a common chirality.
\end{theorem}

\begin{proof}
The proof simply comes by providing an instance where Line formation cannot be accomplished within the provided assumptions. By referring to Figure~\ref{fig:orientation}.b, we note that even if particles share chirality, they are all indistinguishable. No matter the algorithm designed for solving Line formation, an adversary may activate all particles synchronously so that they all behave symmetrically to each other. Hence, any action performed by a particle will be applied by all of them in a symmetric way. It means that any reachable configuration maintains the initial symmetry. Since a configuration solving Line formation for the provided instance requires to distinguish a particle which lies between the other two, we conclude that such a solution cannot be achieved.\qed
\end{proof}

Note that, the arguments provided in the proof of Theorem~\ref{th:imposs} can be extended to any configuration where the initial symmetry is `not compatible' with the formation of a line.

Motivated by Theorem~\ref{th:imposs}, we assume a common orientation to the particles. Consequently, each particle can enumerate its neighborhood, up to distance of 2 hops, as shown in Figure~\ref{fig:orientation}.c. This will be useful for the definition of the resolution algorithm. Actually, it remains open whether it is possible to design an algorithm even when particles share just one direction instead of the full orientation. 

\section{Algorithm \algo}\label{sec:algo}
The rationale behind the name \algo of the proposed algorithm comes by the type of movements allowed. In fact, the evolution of the system on the basis of the algorithm mimics the behavior of particles that fall down 
like drops of rain subject to a westerly wind. 
The Line formation is then reached on the lower part of the initial configuration where there is at least a particle -- what we have called \emph{floor}. 

In order to define the resolution Algorithm \algo, we need to define some functions, expressing properties related to a node of the grid. We make use of the enumeration shown in Figure~\ref{fig:orientation}.c, and in particular to the neighbors enclosed by the two trapezoids.

\begin{definition}
Given a node $v$, the next Boolean functions are defined:
\begin{itemize}
\item \emph{Upper}$(v)$ is \emph{true} if at least one of the visible neighboring nodes from $v$ at positions $\{1,2,4,5,6\}$ is occupied by a particle;
\item \emph{Lower}$(v)$ is \emph{true} if at least one of the visible neighboring nodes from $v$ at positions $\{13,14,15,17,18\}$ is occupied by a particle;
\item \emph{Pointed}$(v)$ is \emph{true} if there exists a particle $p$ occupying a node $u\in N(v)$ such that $p$ is {\expd} along edge $(u,v)$;
\item \emph{Near}$(v)$ is \emph{true} if there exists an empty node $u \in N(v)$ such that \emph{Pointed}$(u)$ is true.
\end{itemize}	
\end{definition}

For the sake of conciseness, sometimes we make use of the above functions by providing a particle $p$ as input in place of the corresponding node $v$ occupied by $p$.

We are now ready to formalize our Algorithm \algo.

\begin{algorithm}[h]
\caption{{\algo}. \label{algo}}
\begin{algorithmic}[1]
{\small
\Require{Node $v$ occupied by a {\cont} particle $p$.}
\Ensure{Line formation.}
	\If{$\neg$\emph{Near}$(v)$}
	    \If{\emph{Pointed}$(v)$}
		\State $p$ expands toward \E\label{algo:move1}
		\Else{
	       \If{$\neg$\emph{Upper}$(v)$ $\wedge$ \emph{Lower}$(v)$}
	       \State $p$ expands toward \SE\label{algo:move2}
           \EndIf}
        \EndIf
   	\EndIf
}
\end{algorithmic}
\end{algorithm}

It is worth noting that Algorithm \algo allows only two types of expansion, toward \E or \SE. Moreover, the movement toward \E can happen only when the node $v$ occupied by a particle is intended to be reached by another particle, i.e., \emph{Pointed}$(v)$ holds. Another remarkable property is that the algorithm only deals with expansion actions. This is due to the constraint of the \silbot model that does not permit to intervene on {\expd} particles, committed to terminate their movement. An example of execution of \algo starting from the configuration of Figure~\ref{fig:es1}.a is shown in the next section.

\section{Running example}\label{sec:run}
In this section, we show a possible execution of Algorithm \algo, 
starting from the configuration shown in Figure~\ref{fig:es1}.a (or equivalently by starting directly from the configuration shown in Figure~\ref{fig:es2}.a). 
Being in an asynchronous setting, there are many possible executions that could occur. In our example, we consider the case where all the particles that can move according to the algorithm apply the corresponding rule. It is basically an execution subject to the fully synchronous schedule (which is a special case of \async).

From the considered configuration of Figure~\ref{fig:es1}.a, Algorithm \algo allows only the particle on top to move. In fact, considering the node $v$ occupied by such a particle, we have that \emph{Near}$(v)$, \emph{Pointed}$(v)$ and \emph{Upper}$(v)$ are all \emph{false}, whereas \emph{Lower}$(v)$ is true. Note that, none of the nodes occupied by the other particles imply function \emph{Upper} to be true but the leftmost for which function \emph{Lower} is false. Hence, the configuration shown in Figure~\ref{fig:es1}.b is reached, eventually. After the movement of the {\expd} particle, see Figure~\ref{fig:es2}.a, the configuration is basically like an initial one with {\cont} and connected particles. The only movement occurring in initial configurations is given by Line~\ref{algo:move2} of Algorithm \algo. In fact, when there are no {\expd} particles, only Line~\ref{algo:move2} can be activated, as Line~\ref{algo:move1} requires function \emph{Pointed} to be true for a node occupied by a {\cont} particle. From the configuration of Figure~\ref{fig:es2}.a, there are two particles -- the top ones, that can move according to the algorithm. If both are activated, configuration of Figure~\ref{fig:es2}.b is obtained.
    Successively, the rightmost {\expd} particle is free to move, whereas the other {\expd} particle allows the pointed particle to expand, as shown in Figure~\ref{fig:es2}.c, by means of Line~\ref{algo:move1} of the algorithm.

\begin{figure}[ht]
        \centering
  a) \includegraphics[width=3.5cm]{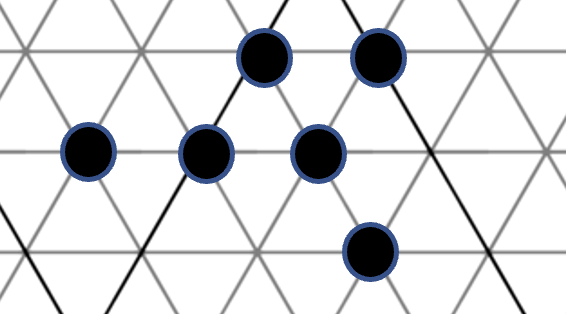}  
		\hspace{.5mm}        
  b) \includegraphics[width=3.5cm]{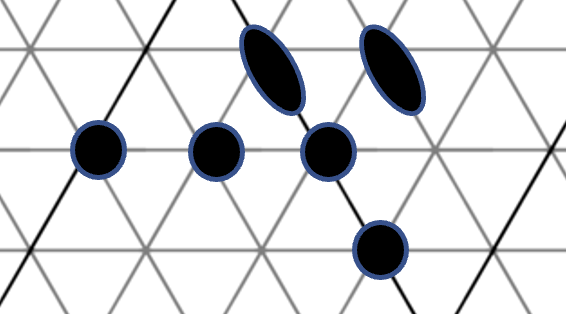}
          \hspace{.5mm}    
  c) \includegraphics[width=3.5cm]{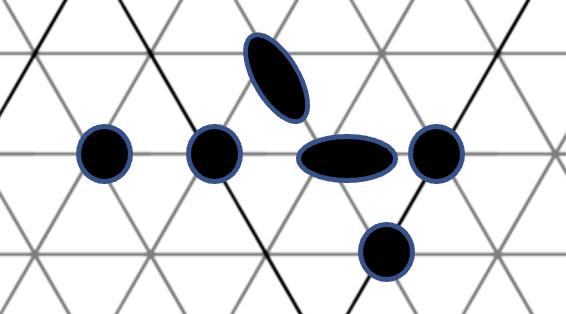}\\ \ \\ \  \\
  d)  \includegraphics[width=3.5cm]{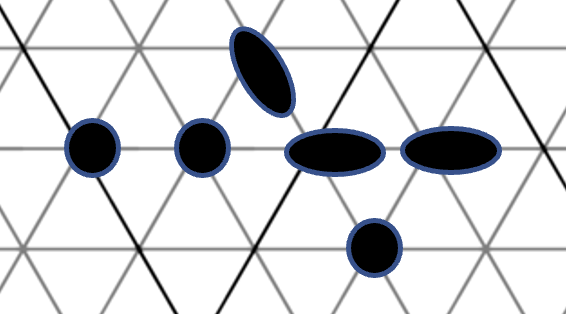}
          \hspace{1mm}       
  e) \includegraphics[width=3.5cm]{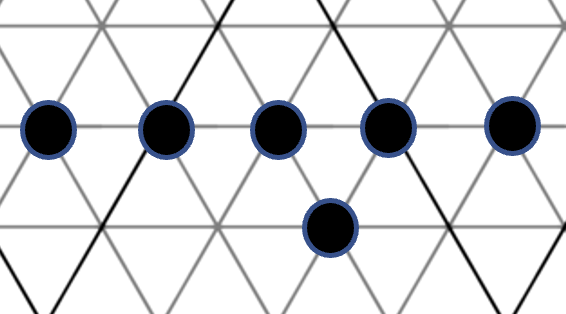}
  		\hspace{1mm}        
  f) \includegraphics[width=3.5cm]{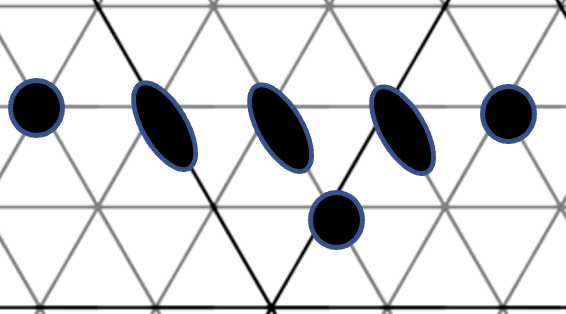}\\ \ \\ \ \\
  g)  \includegraphics[width=3.5cm]{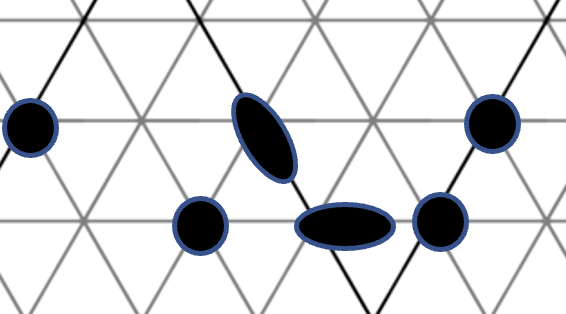}  
		\hspace{1mm}        
  h) \includegraphics[width=3.5cm]{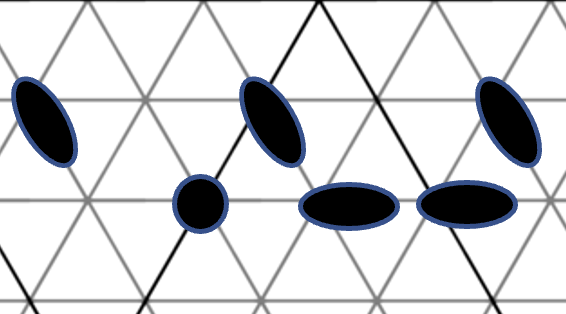}
          \hspace{1mm}        
  i) \includegraphics[width=3.5cm]{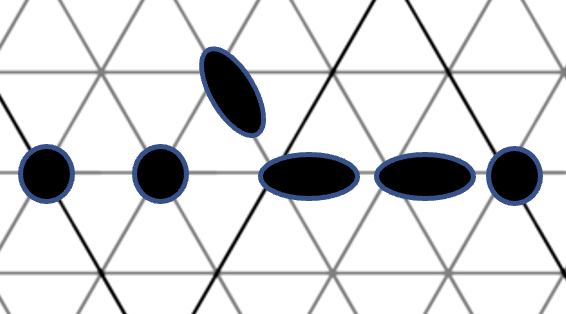}\\ \ \\ \ \\
    		\hspace{1mm}
  j)  \includegraphics[width=4cm]{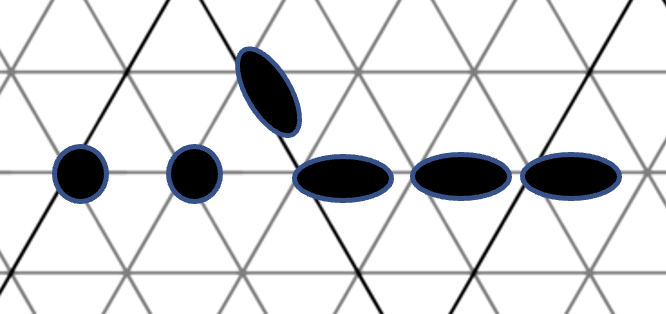}     
    		\hspace{1mm}     
  k) \includegraphics[width=4cm]{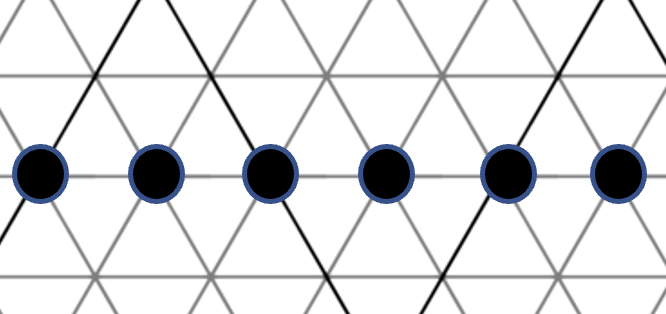}
        \caption{\small A possible execution when starting from the configuration shown in Figure~\ref{fig:es1}.a.}
        \label{fig:es2}
\end{figure}

As already observed, the movement toward \SE is generated by the rule at Line~\ref{algo:move2} of Algorithm \algo, whereas the movement toward \E can only be induced by {\expd} particles as specified by the rule at Line~\ref{algo:move1}. By keep applying the two rules among all particles, the execution shown in the subsequent Figures~\ref{fig:es2}.d--k is obtained, hence leading to the configuration where all particles are {\cont} and aligned along \emph{floor}. It is worth noting that the configuration shown in Figure~\ref{fig:es2}.g is disconnected.  However, as we are going to show, the possible disconnections occurring during an execution are always recovered. In particular, in the specific example, connectivity is recovered right after as shown in Figure~\ref{fig:es2}.i.

\section{Correctness and Optimality}\label{sec:corr}
In this section, we prove the correctness of Algorithm \algo as well as its optimality in terms of number of moves performed by the particles.

We prove the correctness of Algorithm {\algo} by showing that the four following claims hold:
\begin{description}
\item[Claim 1 - Configuration Uniqueness.] Each configuration generated during the execution of the algorithm is unique, i.e., non-repeatable, after movements, on the same nodes nor on different nodes;
\item[Claim 2 - Limited Dimension.] The extension of any (generated) configuration is confined within a finite bounding box of sides $O(n)$;
\item[Claim 3 - Evolution guarantee.] If the (generated) configuration is connected and not final there always exists at least a particle that can expand or contract;
\item[Claim 4 - Connectivity.] If two particles initially neighboring to each other get disconnected, they recover their connection sooner or later (not necessarily becoming neighbors).
\end{description}
 
The four claims guarantee that a final configuration is achieved, eventually, in finite time, i.e., Line formation is solved. In fact, by assuming the four claims true, we can state the next theorem.

\begin{theorem}
\label{th:correctness}
Given $n$ {\cont} particles forming a connected configuration, Algorithm {\algo} terminates in a connected configuration where all the particles are aligned along \emph{floor}.
\end{theorem}

\begin{proof} By Claim 3 we have that from any non-final configuration reached during an execution of \algo there is always at least one particle that moves. Hence, by Claim 1, any subsequent configuration must be different from any already reached configuration. However, since Claim 2 states that the area where the particles move is limited, then a final configuration must be reached as the number of achievable configurations is finite. 
Actually, if we imagine a configuration made of disconnected and {\cont} particles, all lying on \emph{floor}, then the configuration is not final according to Definition~\ref{def:final} but none of the particles would move. However, by Claim 4, we know that such a type of configurations cannot occur, and in particular, if two particles initially neighboring to each other get disconnected, then they recover their connection, eventually. Since the initial configuration is connected, then we are ensured that also the final configuration is connected as well.\qed
\end{proof}

We now provide a proof for each of the above claims.

\begin{proof}[of Claim 1 - Configuration Uniqueness]

Since the movements allowed by the algorithm are toward either \E or \SE only, then the same configuration on the same nodes cannot arise during an execution as it would mean that some particles have moved toward \W, \NW, or \NE. Concerning the case to form the same configuration but on different nodes, it is sufficient to note that a particle lying on a node $v$ of \emph{floor} can only move toward \E (since \emph{Lower}$(v)$ is false, cf. Line~\ref{algo:move2} of Algorithm \algo).
Hence, either none of the particles on \emph{floor} move, in which case the same configuration should appear on the same nodes -- but this has been already excluded; or the same configuration may appear if all the particles move toward \E. However, based on the algorithm, the only movement that can occur from an initial configuration is toward \SE, hence the claim holds.	\qed
\end{proof}

\begin{proof}[of Claim 2 - Limited Dimension]

From the arguments provided to prove Claim 1, we already know that any configuration obtained during an execution of \algo never overpasses \emph{floor}, defined by the initial configuration. Moreover, since the movements are toward either \E or \SE only, then the northern and the western sides of the bounding box of the initial configuration are never overpassed as well. Concerning the eastern side, we show that this can be shifted toward east in the generated configurations at most $n$ times

About movements toward \SE that overpass the eastern side, they cannot happen more than $n-1$ times according to Algorithm \algo. In fact, each time it happens, the northern side moves toward south.

About the movement toward \E, it requires a pushing-like process by another particle that either comes from \W or from \NW. The claim then follows by observing that a particle can be pushed at most $n-1$ times, one for each other particle. In fact, if a particle $p$ is pushed toward \E, then the pushing particle $p'$ either comes from \W or from \NW, i.e., after the pushing $p$ and $p'$ are on the same WestEast axis. Hence, in order to push again $p$ toward \E, it is necessary that a third particle, $p''$ pushes $p'$ that in turn pushes $p$.
This may happen, for instance, if initially the particles are all aligned along the western side of the bounding box. Hence, by making the union of the bounding boxes of all the configurations obtained during an execution of \algo, the obtained box has the sides of size upper bounded by $n$. \qed
\end{proof}

\begin{proof}[of Claim 3 - Evolution guarantee]

Let us assume the configuration does contain a particle $p$, occupying node $v$, {\expd} toward node $u$. If $u$ is empty, then $p$ (or possibly another particle) will reach $u$, eventually. If $u$ is occupied, then the particle $p'$ in $u$ -- if not already {\expd}, will be pushed to move toward \E. In any case, there must be a particle at the end of a chain of {\expd} particles that either expands itself or moves toward the empty node toward which it is expanded. In any case, the configuration evolves.

Let us consider then the case where all the particles are {\cont} and connected.
If all the particles lie on \emph{floor}, then the configuration is final. Hence, if the configuration is not final, there must exist a particle $p$ occupying a node $v$ which is not on \emph{floor} such that, $\neg Near(v)~\wedge~\neg Pointed(v)~\wedge~\neg Upper(v)~\wedge~ Lower(v)$ holds, i.e., according to Algorithm \algo, $p$ expands toward \SE. The existence of $p$ is guaranteed by the fact that $\neg Near(v)~\wedge~\neg Pointed(v)$ clearly holds since none of the particles are {\expd}, whereas $\neg Upper(v)~\wedge~ Lower(v)$ holds for at least one of the topmost particles that of course does not admit neighboring particles on top, but admits particles below, due to connectivity. \qed
\end{proof}

\begin{proof}[of Claim 4 - Connectivity]

Let us consider two neighboring particles $p$ and $p'$ of the initial configuration. Without loss of generality, let us assume that the two particles become disconnected due to the movement of $p$ from node $v$ to node $u$. In fact, expansions do not cause disconnections as an {\expd} particle still maintains the node occupied.
If the movement is toward \E, then we are sure there is another particle {\expd} toward $v$, i.e., $v$ remains semi-occupied. Consequently, either $p'$ moves and recovers its connection with $p$ or another particle moves to $v$, again recovering the connection between $p$ and $p'$. Moreover, after its movement, $p$ cannot move again as long as $v$ remains semi-occupied since $Near(p)$ is true during that time; whereas, if $p'$ moves during that time (necessarily toward \E or \SE), it becomes neighbor of $p$ again. 

Then, the movement of $p$ must be toward \SE. According to Algorithm \algo, $p$ has decided to move toward \SE because: $Near(v)$ is false, i.e., none of the nodes in $N(v)$ is semi-occupied; $Pointed(v)$ is false; $Upper(v)$ is false and in particular the are no particles in positions $\{4,5,6\}$ according to the enumeration of its neighborhood shown in Figure~\ref{fig:orientation}.c; whereas there is at least one particle $p''$ among positions $\{13,15,17,18\}$. In fact, $14$ must be empty as $p$ is moving there. Hence, the movement toward $14$ makes $p$ neighboring $p''$. 
    It follows that, if the movement of $p$ has caused a disconnection from $p'$, then $p'$ is in position $9$, with respect to $v$, that represents the connection to $p$ before the movement. In fact, we know that positions $\{5,6\}$ are empty, whereas the movement to $14$ maintains $p$ neighboring with $\{10,13\}$, i.e., only the connection to $9$ can get lost. 
    Hence, $p'$ makes $Upper(p)$ true, and $p$ makes $Lower(p')$ true. It follows that $p$ won't move anymore unless another particle $\overline{p}$ (possibly arriving successively) pushes it from $v$ or from $13$. In either cases, $\overline{p}$ connects $p$ with $p'$. If $p$ doesn't move before $p'$, then $p'$ must move, eventually. In fact, this happens as soon as either it is pushed or the $Upper$ function evaluated from $9$ becomes false. By Claims 1, 2 and 3, this must happen, eventually, since the configuration is not final.\qed
\end{proof}

We are now ready to prove the optimality of Algorithm \algo in terms of number of total moves performed by the robots.

\begin{lemma}
\label{lemma:time}
Given $n$ {\cont} particles forming a connected configuration, Algorithm {\algo} terminates within $O(n^2)$ movements.
\end{lemma}

\begin{proof}
In order to prove the lemma, it suffices to remark that any particle moves at most $n-1$ times toward \E and $n-1$ times toward \SE, hence obtaining a number of total movements upper bounded by $O(n^2)$. \qed
\end{proof}

\begin{theorem}
\label{th:optimality}
Algorithm {\algo} is asymptotically optimal in terms of number of movements.
\end{theorem}

\begin{proof}
As proven in~\cite{DGS15}, Line formation requires $\Omega(n^2)$ movements.
That proof simply comes by assuming the initial configuration formed by $n$ particles composing a connected structure of diameter at most $2\sqrt{n} + 2$ (e.g., if they form a hexagonal or square shape), and then summing up all the necessary movements required to reach a configuration where particles form a line.
Hence, by combining such a result with Lemma~\ref{lemma:time}, the claim holds.\qed
\end{proof}

\section{Conclusion}\label{sec:concl}
We investigated on the Line formation problem within PM on the basis of the \silbot model.
With the aim of considering the smallest set of assumptions, we proved how chirality was not enough for particles to accomplish Line formation. We then endowed particles with a common sense of direction and we proposed \algo, an optimal algorithm -- in terms of number of movements, for solving Line formation. Actually, it remains open whether by assuming just one common direction is enough for solving the problem. Furthermore, although in the original paper about \silbot~\cite{DDDNP20b} it has been pointed out that 1 hop visibility is not enough for solving the Leader Election, it is worth investigating what happens for Line formation.

Other interesting research directions concern the resolution of other basic primitives, the formation of different shapes or the more general pattern formation problem.  Also variants on the original \silbot model deserve main attention. As shown in Table~\ref{tab:problems}, small modifications to the original model may allow the resolution of challenging tasks. It would be interesting, for instance, to understand what might change if {\expd} particles are allowed to revoke from their commitment on moving forward, i.e., if algorithms could deal also with {\expd} particles. 

Furthermore, adding a few bits of visible memory like allowing the particles to assume different states other than {\cont} and {\expd}, or being endowed with visible lights similar to those studied in robot systems as in~\cite{DDFN18}, might reveal higher potentials for PM.

\bibliography{bib}

\end{document}